\documentclass[submission,copyright,creativecommons]{eptcs}
\providecommand{\event}{TARK 2023} 

\usepackage{iftex}

\ifpdf
  \usepackage{underscore}         
  \usepackage[T1]{fontenc}        
\else
  \usepackage{breakurl}           
\fi

\usepackage{url}
\usepackage[utf8]{inputenc}
\usepackage[small]{caption}
\usepackage{graphicx}
\usepackage{mathrsfs}
\usepackage{amsmath}
\usepackage{amsthm}
\usepackage{amssymb}
\usepackage{booktabs}
\usepackage{algorithm}
\usepackage{algorithmic}
\usepackage[switch]{lineno}
\usepackage{caption}
\usepackage{subcaption}

\usepackage{balance} 
\usepackage{comment}
\usepackage{thm-restate}

\newtheorem{example}{Example}
\newtheorem{theorem}{Theorem}
\newtheorem{proposition}{Proposition}

\newcommand{\putaway}[1]{}

\def\prefeq{\succeq}
\def\pref{\succ}

\def\_{\mbox{-}}
\def\~{\mbox{$\sim$}}

\def\set#1{\{#1\}}		
\def\x{\times}			
\def\to{\rightarrow}		
\let\.=\cdot
\def\Not{\neg}			
\def\Or{\vee}			
\def\Implies{\rightarrow}	
\def\Iff{\leftrightarrow}	

\def\Union{\bigcup}		

\def\Nat{\mathbb{N}}                
\def\Real{\mathbb{R}}               
\def\phi{\varphi}

\def\TEMPORAL#1{\mbox{\small\boldmath$\mathbf{#1}$}}

\def\ltlnext{\TEMPORAL{X}}
\def\sometime{\TEMPORAL{F}} 
\def\always{\TEMPORAL{G}}
\def\until{\,\TEMPORAL{U}\,}

\def\eventually{\sometime}
\def\calA{\mathcal{A}}  \def\calC{\mathcal{C}}
  
\def\calG{\mathcal{G}}  
  \def\calL{\mathcal{L}}
 \def\calN{\mathcal{N}} 
  
\def\calS{\mathcal{S}} \def\calT{\mathcal{T}} 
\def\calV{\mathcal{V}}  
 
\usepackage{tikz}
\usepackage{subcaption}
\usepackage{booktabs}
\usepackage{enumitem}
\usepackage{ifthen}
\usepackage{comment}
\usepackage{bm}
\usetikzlibrary{arrows,automata,positioning,calc}

\newcommand{\ac}{\mathit{Ac}}
\newcommand{\vac}{\vec{\ac}}
\newcommand{\acp}{\vec{\alpha}}
\newcommand{\strat}{\vec{\sigma}}
\newcommand{\nasheq}{\textsc{ne}}

\newcommand{\ani}{\textsc{ani}(\calG,\Upsilon)}
\newcommand{\eni}{\textsc{eni}(\calG,\Upsilon)}

\newcommand{\indev}[1]{\textsc{in}(#1,\Gamma)}

\title{Incentive Engineering for Concurrent Games}
\author{David Hyland 
\institute{University of Oxford\\ Oxford, United Kingdom}
\email{david.hyland@cs.ox.ac.uk}
\and
Julian Gutierrez
\institute{Monash University\\
Melbourne, Australia}
\email{julian.gutierrez@monash.edu}
\and
Michael Wooldridge 
\institute{University of Oxford\\ Oxford, United Kingdom}
\email{mjw@cs.ox.ac.uk}
}
\newcommand{\titlerunning}{Incentive Engineering for Concurrent Games}
\newcommand{\authorrunning}{D. Hyland, J. Gutierrez, \& M.J. Wooldridge}

\begin{document}
\maketitle

\begin{abstract}
We consider the problem of incentivising desirable behaviours in multi-agent systems by way of taxation schemes. Our study employs the concurrent games model: in this model, each agent is primarily motivated to seek the satisfaction of a goal, expressed as a Linear Temporal Logic (LTL) formula; secondarily, agents seek to minimise costs, where costs are imposed based on the actions taken by agents in different states of the game. In this setting, we consider an external principal who can influence agents' preferences by imposing taxes (additional costs) on the actions chosen by agents in different states. The principal imposes taxation schemes to motivate agents to choose a course of action that will lead to the satisfaction of their goal, also expressed as an LTL formula. However, taxation schemes are limited in their ability to influence agents' preferences: an agent will always prefer to satisfy its goal rather than otherwise, no matter what the costs. The fundamental question that we study is whether the principal can impose a taxation scheme such that, in the resulting game, the principal's goal is satisfied in at least one or all runs of the game that could arise by agents choosing to follow game-theoretic equilibrium strategies. We consider two different types of taxation schemes: in a \emph{static} scheme, the same tax is imposed on a state-action profile pair in all circumstances, while in a \emph{dynamic} scheme, the principal can choose to vary taxes depending on the circumstances. We investigate the main game-theoretic properties of this model as well as the computational complexity of the relevant decision problems.
\end{abstract}

\section{Introduction}
\emph{Rational verification} is the problem of establishing which temporal logic properties will be satisfied by a multi-agent system, under the assumption that agents in the system choose strategies that form a game-theoretic equilibrium~\cite{fisman2010rational,DBLP:conf/aaai/WooldridgeGHMPT16,DBLP:journals/ai/GutierrezNPW20}. Thus, rational verification enables us to verify which desirable and undesirable behaviours could arise in a system through individually rational choices. This article, however, expands beyond verification and studies methods for incentivising outcomes with favourable properties while mitigating undesirable consequences. One prominent example is the implementation of Pigovian taxes, which effectively discourage agents from engaging in activities that generate negative externalities. These taxes have been extensively explored in various domains, including sustainability and AI for social good, with applications such as reducing carbon emissions, road congestion, and river pollution \cite{mankiw2009smart, lyon2001pigouvian, pigou2017economics}.

We take as our starting point the work of~\cite{wooldridge:2013a}, who considered the possibility of influencing one-shot Boolean games by introducing taxation schemes, which impose additional costs onto a game at the level of individual actions. In the model of preferences considered in~\cite{wooldridge:2013a}, agents are primarily motivated to achieve a goal expressed as a (propositional) logical formula, and only secondarily motivated to minimise costs. This logical component limits the possibility to influence agent preferences: an agent can never be motivated by a taxation scheme away from achieving its goal. In related work, Wooldridge et al. defined the following implementation problem: given a game $G$ and an objective $\Upsilon$, expressed as a propositional logic formula, does there exists a taxation scheme $\tau$ that could be imposed upon $G$ such that, in the resulting game $G^\tau$, the objective~$\Upsilon$ will be satisfied in at least one Nash equilibrium \cite{wooldridge:2013a}. 

We develop these ideas by applying models of finite-state automata to introduce and motivate the use of history-dependent incentives in the context of \emph{concurrent games}~\cite{alur:2002a}. In a concurrent game, play continues for an infinite number of rounds, where at each round, each agent simultaneously chooses an action to perform. Preferences in such a multiplayer game are defined by associating with each agent $i$ a Linear Temporal Logic (LTL) goal $\gamma_i$, which agent $i$ desires to see satisfied. In this work, we also assume that actions incur costs, and that agents seek to minimise their limit-average costs.

Since, in contrast to the model of  \cite{wooldridge:2013a}, play in our games continues for an infinite number of rounds, we find there are two natural variations of taxation schemes for concurrent games. In a \emph{static} taxation scheme, we impose a fixed cost on state-action profiles so that the same state-action profile will always incur the same tax, no matter when it is performed. In a \emph{dynamic} taxation scheme, the same state-action profile may incur different taxes in different circumstances: it is history-dependent. We first show that dynamic taxation schemes are strictly more powerful than static taxation schemes, making them a more appropriate model of incentives in the context of concurrent games, characterise the conditions under which an LTL objective $\Upsilon$ can be implemented in a game using dynamic taxation schemes, and begin to investigate the computational complexity of the corresponding decision problems.

\section{Preliminaries}
\label{sec:prelim}
Where $S$ is a set, we denote the powerset of $S$ by $\mathbf{2}^S$.
We use various propositional languages to express properties of the
systems we consider. In these languages, we will let $\Phi$ be a
finite and non-empty vocabulary of Boolean variables, with typical
elements $p, q, \ldots$. Where $a$ is a finite word and $b$ is also a word (either
finite or infinite), we denote the word obtained by concatenating
$a$ and $b$ by $a b$. Where $a$ is a finite word, we denote by $a^\omega$ the infinite repetition of $a$. Finally, we use $\Real_n^+$ for the set of $n$-tuples of non-negative real numbers.

\vspace*{1ex}\noindent\textbf{Concurrent Game Arenas:}
  We work with concurrent game structures, which in this work we will refer to as \emph{arenas} (to distinguish them from the game structures that we introduce later in this section)~\cite{alur:2002a}.
Formally a \textit{concurrent game arena} is given by a structure \[\calA = (\calS,\calN,\ac_1, \ldots,\ac_n,\calT,\calC,\calL,s_0),\] where: $\calS$ is a finite and non-empty set of \textit{arena states}; $\calN = \set{1, \ldots, n}$ is the set of \textit{agents} -- for any $i \in \calN$, we let $-i = \calN \setminus \set{i}$ denote the set of \textit{all agents excluding} $i$; for each $i \in \calN$, $\ac_i$ is the finite and non-empty set of unique actions available to agent $i$ -- we let $\ac = \Union_{i \in \calN} \ac_i$ denote the set of all \textit{actions} available to all players in the game and $\vac = \ac_1 \x \cdots \x \ac_n$ denote the set of all \textit{action profiles}; $\calT : \calS \x \ac_1 \x \cdots \x \ac_n \to \calS$ is the \textit{state transformer function} which prescribes how the state of the arena is updated for each possible action profile -- we refer to a pair $(s,\acp)$, consisting of a state $s \in \calS$ and an action profile $\acp \in \vac$ as a \textit{state-action profile}; $\calC : \calS \x \ac_1 \x \cdots \x \ac_n \to \Real_+^n$ is the \textit{cost function} -- given a state-action profile $(s,\acp)$ and an agent $i \in \calN$, we write $\calC_i(s,\acp)$ for the $i$-th component of $\calC(s,\acp)$, which corresponds to the cost that agent $i$ incurs when $\acp$ is executed at $s$; $\calL : \calS \to \mathbf{2}^\Phi$ is a \textit{labelling function} that specifies which propositional variables are true in each state $s \in \calS$; and $s_0 \in \calS$ is the \textit{initial state} of the arena. In what follows, it is useful to define for every agent $i \in \calN$ the value $c_i^*$ to be the maximum cost that~$i$ could incur through the execution of a state-action profile: $c_i^* = \max \set{\calC_i(s,\acp) \mid s \in \calS, \acp \in \vac}$.

\vspace*{1ex}\noindent\textbf{Runs:}
Games are played in an arena as follows. The arena begins in its initial state $s_0$, and each agent $i \in \calN$ selects an action $\alpha_i \in \ac_i$ to perform; the actions so selected define an action profile, $\acp \in \ac_1 \x \cdots \x \ac_n$. The arena then transitions to a new state $s_1 = \calT(s_0,\alpha_1, \ldots,\alpha_n)$. Each agent then selects another action $\alpha'_i \in Ac_i$, and the arena again transitions to a new state $s_2 = \calT(s_1,\acp')$. In this way, we trace out an infinite interleaved sequence of states and action profiles, referred to as a \emph{run}, $\rho : s_0
  \stackrel{\acp_0}{\longrightarrow} s_1
  \stackrel{\acp_1}{\longrightarrow} s_2
  \stackrel{\acp_2}{\longrightarrow} \cdots$.
  
Where $\rho$ is a run and $k \in \Nat$, we write $s(\rho,k)$ to denote the state indexed by $k$ in $\rho$, so $s(\rho,0)$ is the first state in $\rho$, $s(\rho,1)$ is the second, and so on. In the same way, we denote the $k$-th action profile played in a run $\rho$ by $\acp(\rho,k-1)$ and to single out an individual agent $i$'s $k$-th action, we write $\alpha_i(\rho,k-1)$.

Above, we defined the cost function $\calC$ with respect to
individual state-action pairs. In what follows, we find it useful to lift the
cost function from individual state-action pairs to sequences of state-action pairs and runs. Since runs are infinite, simply taking the sum of costs is not appropriate: instead, we consider the cost of a run to be the \emph{average} cost incurred by an agent $i$ over the run; more precisely, we define the average cost incurred by agent $i$ over the first $t$ steps of the run $\rho$ as $\calC_i(\rho,0:t) = \frac{1}{t+1}\sum_{j=0}^t\calC_i(\rho,j)$ for $t\geq 1$, whereby $\calC_i(\rho,j)$ we mean $\calC_i(s(\rho,j),\acp(\rho,j))$. 
Then, we define the cost incurred by an agent $i$ over the run $\rho$, denoted $\calC_i(\rho)$, as the \emph{inferior limit of means}: $\calC_i(\rho) =\liminf_{t\to\infty} \calC_i(\rho,0:~t)$. It can be shown that the value $\calC_i(\rho)$ always converges because the sequence of averages $\calC_i(\rho,0:t)$ is Cauchy.

\vspace*{1ex}\noindent\textbf{Linear Temporal Logic:}
We use the language of Linear Temporal Logic
(LTL) to express properties of runs~\cite{pnueli:77a,emerson:90a}. Formally, the syntax of LTL is defined wrt.\ a set $\Phi$
of Boolean variables by the following grammar:
    \begin{eqnarray} 
    \varphi ::= 
        \top \mid
    	p \mid		
    	\neg \varphi \mid
    	\varphi \vee \varphi \mid
    	\ltlnext \varphi \mid
    	\varphi \until \varphi 
    \end{eqnarray}
where $p \in \Phi$. 
Other usual logic connectives (``$\bot$'', ``$\And$'', ``$\Implies$'', ``$\Iff$'') are defined in terms of~$\Not$ and~$\Or$ in the conventional way. Given a set of
variables~$\Phi$, let~$LTL(\Phi)$ be the set of LTL formulae
over~$\Phi$; where the variable set~$\Phi$ is clear from the context, we simply write $LTL$. We interpret formulae of LTL with respect to pairs~$(\rho,t)$, where~$\rho$ is a run, and $t \in \Nat$ is a temporal index
into~$\rho$. Any given LTL formula may be true at none or multiple
time points on a run; for example, a formula $\ltlnext q$ will be true
at a time point $t\in\Nat$ on a run $\rho$ if $q$ is true on a run
$\rho$ at time $t+1$. We will write $(\rho,t)\models\phi$ to mean that
$\phi \in LTL$ is true at time $t \in \Nat$ on run $\rho$.  The rules
defining when formulae are true (i.e., the semantics of LTL) are
defined as follows:
%
\[
\renewcommand{\arraystretch}{1.2}
\begin{array}{lcl}
(\rho,t)\models \top & & 
\\
(\rho,t)\models p & \mbox{ iff } & p\in\calL(s(\rho,t))
 \quad \mbox{(where $p  \in \Phi$)} 
\\
(\rho,t)\models \neg \varphi & \mbox{iff} & \mbox{it is not the case that } (\rho,t) \models \varphi\\
(\rho,t)\models \varphi \vee \psi & \mbox{iff} & (\rho,t) \models \varphi 
    \mbox{ or } (\rho,t) \models \psi\\
(\rho,t)\models \ltlnext\varphi & \mbox{iff} & (\rho,t+1) \models \varphi\\
(\rho,t)\models \varphi\until\psi & \mbox{iff} & \mbox{for some } t' \geq t : \ (\rho,t') \models \psi \mbox{ and}\\
		&& \mbox{for all } t \leq t'' < t': (\rho,t'') \models \varphi 
\end{array}
\]

We write $\rho \models \phi$ as a shorthand for $(\rho,0) \models \phi$, in which case we say that \emph{$\rho$ satisfies~$\varphi$}. A formula $\varphi$ is \emph{satisfiable} if there is some run satisfying~$\varphi$. Checking satisfiability for LTL formulae is known to be 
\textsc{PSpace}-complete~\cite{sistla:85b}, while the synthesis problem for LTL is \textsc{2ExpTime}-complete~\cite{pnueli:89a}. 
In addition to the LTL tense operators $\ltlnext$ (``in the next state\ldots'') and $\until$ (``\ldots until \ldots''), we make use of the two derived operators $\sometime$ (``eventually\ldots'') and $\always$ (``always\ldots''), which are defined as follows~\cite{emerson:90a}: 
$\sometime\varphi = \mathop\top \until \varphi$ and $\always\varphi = \neg\sometime\neg\varphi$. 

\vspace*{1ex}\noindent\textbf{Strategies:} 
We model strategies for agents as finite-state machines with output. Formally,  strategy $\sigma_i$ for agent $i \in \calN$ is given by a structure $\sigma_i = (Q_i,next_i,do_i,q_i^0)$, where $Q_i$ is a finite set of machine states, $next_i : Q_i \x \ac_1 \x \cdots \x \ac_n \to Q_i$ is the machine's state transformer function, $do_i : Q_i \to \ac_i$ is the machine's action selection function, and $q^0_i \in Q_i$ is the machine's initial state. A collection of strategies, one for each agent $i\in\calN$, is a \emph{strategy profile}: $\strat = (\sigma_1, \ldots, \sigma_n)$.  A strategy profile $\strat$ enacted in an arena $\calA$ will generate a unique run, which we denote by $\rho(\strat,\calA)$; the formal definition is standard, and we will omit it here \cite{DBLP:journals/ai/GutierrezNPW20}. Where $\calA$ is clear from the context, we will simply write $\rho(\strat)$. For each agent $i\in\calN$, we write $\Sigma_i$ for the set of all possible strategies for the agent and $\Sigma = \Sigma_1 \x \cdots \x \Sigma_n$ for the set of all possible strategy profiles for all players.

For a set of distinct agents $A \subseteq \calN$, we write $\Sigma_A = \prod_{i \in A} \Sigma_i$ for the set of partial strategy profiles available to the group $A$ and $\Sigma_{-A} = \prod_{j \in \calN \setminus A} \Sigma_j$ for the set of partial strategy profiles available to the set of all agents excluding those in $A$. Where $\strat = (\sigma_1, \ldots, \sigma_i, \ldots, \sigma_n)$ is a strategy profile and $\sigma_i'$ is a strategy for agent $i$, we denote the strategy profile obtained by replacing the $i$-th component of $\strat$ with $\sigma_i'$ by $(\strat_{-i},\sigma_i')$. Similarly, given a strategy profile $\strat$ and a set of agents $A \subseteq \calN$, we write $\strat_A = (\sigma_i)_{i \in A}$ to denote a partial strategy profile for the agents in $A$ and if $\strat_A' \in \Sigma_A$ is another partial strategy profile for $A$, we write $(\strat_{-A},\strat_A')$ for the strategy profile obtained by replacing $\strat_A$ in $\strat$ with $\strat_A'$.
 
\vspace*{1ex}\noindent\textbf{Games, Utilities, and Preferences:}
We obtain a \emph{concurrent game} from an arena $\calA$ by associating with each agent $i$ a goal $\gamma_i$, represented as an LTL formula. Formally, a concurrent game $\calG$ is given by a structure \[\calG = (\calS,\calN,\ac_1, \ldots,\ac_n,\calT,\calC,\calL,s_0, \gamma_1, \ldots, \gamma_n),\] where $(\calS,\calN,\ac_1, \ldots,\ac_n,\calT,\calC,\calL,s_0)$ is a concurrent game arena, and $\gamma_i$ is the LTL goal of agent $i$, for each $i \in \calN$. Runs in a concurrent game $\calG$ are defined over the game's arena $\calA$, and hence we use the notations $\rho(\strat,\calG)$ and $\rho(\strat,\calA)$ interchangeably. When the game or arena is clear from the context, we omit the $\calG$ and simply write $\rho(\strat)$. Given a strategy profile $\strat$, the generated run $\rho(\strat)$ will satisfy the goals of some agents and not satisfy the goals of others, that is, there will be a set $W(\strat) = \set{i \in \calN : \rho(\strat)\models \gamma_i}$ of \textit{winners} and a set $L(\strat) = \calN \setminus W(\strat)$ of \textit{losers}.

We are now ready to define preferences for agents. Our basic idea is that, as in~\cite{wooldridge:2013a}, agents' preferences are structured: they first desire to accomplish their goal, and secondarily desire to minimise their costs. To capture this idea, it is convenient to define preferences via utility functions $u_i$ over runs, where $i$'s utility for a run $\rho$ is
\[ u_i(\rho) = 
\left\{\begin{array}{ll}
1+c_i^*-\calC_i(\rho)   & \mbox{if $\rho\models \gamma_i$}\\
-\calC_i(\rho)          & \mbox{otherwise.}
\end{array}\right.
\]

Defined in this way, if an agent $i$ gets their goal achieved, their utility will lie in the range $[1,c_i^* + 1]$ (depending on the cost she incurs), whereas if she does not achieve their goal, then their utility will lie within $[-c_i^*,0]$.
Preference relations $\prefeq_i$ over runs are then defined in the obvious way: $\rho_1 \prefeq_i \rho_2$ if and only if $u_i(\rho_1) \geq u_i(\rho_2)$, with indifference relations $\sim_i$ and strict preference relations $\pref_i$ defined as usual.

\begin{figure}
    \centering
    \begin{subfigure}{.49\textwidth}
    \centering
        \includegraphics[scale=0.4]{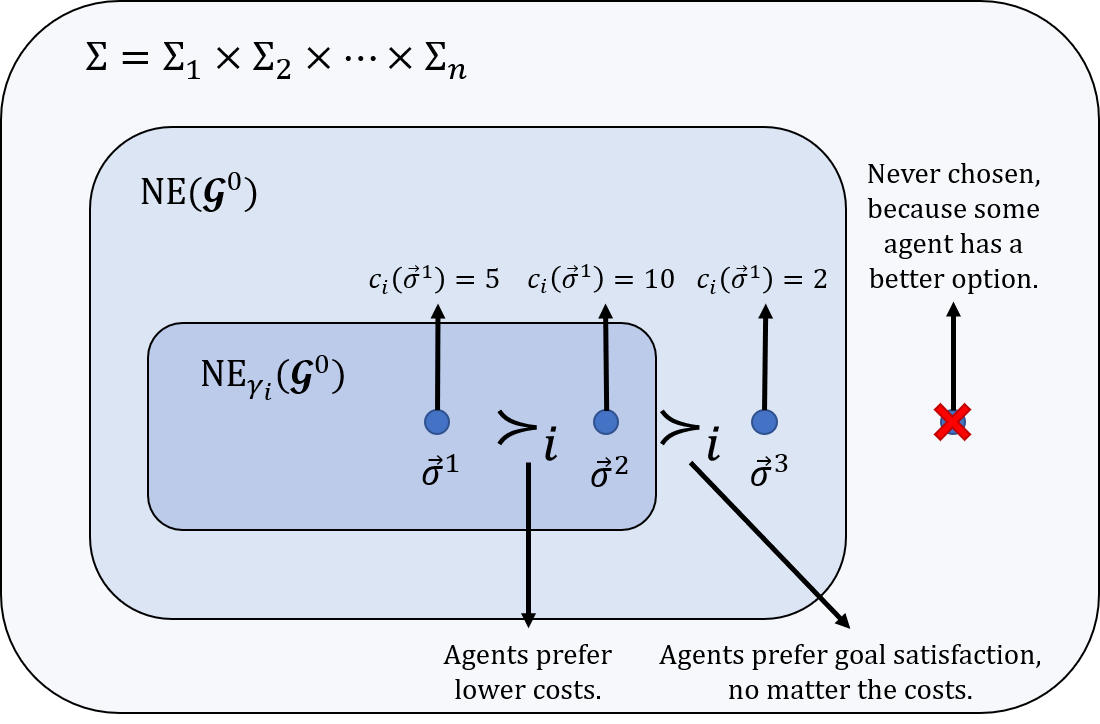}
        \caption{}
        \label{fig:agentprefs}
    \end{subfigure}
    \begin{subfigure}{.5\textwidth}
        \centering
        \includegraphics[scale=0.55]{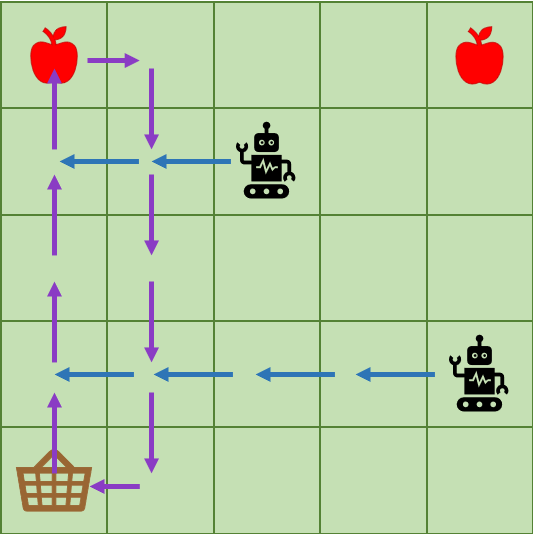}
        \caption{}
        \label{example:applerobots}
    \end{subfigure}
    \caption{(a) Illustration of the lexicographic quantitative and qualitative preferences of agents. (b) A concurrent game where two robots are situated in a grid world and are programmed to 1) never crash into another robot and 2) to secondarily minimise their limit-average costs. The arrows indicate how a run may be decomposed into a non-repeating and an infinitely-repeating component.}
\end{figure}

\vspace*{1ex}\noindent\textbf{Nash equilibrium:}
 A strategy profile $\strat$ is a (pure strategy) Nash equilibrium if there is no agent $i$ and strategy $\sigma_i'$ such that $\rho(\strat_{-i},\sigma_i') \pref_i \rho(\strat)$. If such a strategy $\sigma_i'$ exists for a given agent $i$, we say that $\sigma_i'$ is a \textit{beneficial deviation} for $i$ from $\strat$.
 Given a game $\calG$, let $\nasheq(\calG)$ denote its set of Nash equilibria. In general, Nash equilibria in this model of concurrent games may require agents to play infinite memory strategies \cite{chatterjee:2005a}, but we do not consider these in this study \footnote{Even in the purely quantitative setting where all agents' goals are $\top$, it is still possible that some Nash equilibria require infinite memory \cite{gutierrez2019equilibrium}.}.
 Where $\phi$ is an LTL formula, we find it useful to define $\nasheq_\phi(\calG)$ to be the set of Nash equilibrium strategy profiles that result in $\phi$ being satisfied: $\nasheq_\phi(\calG)= \set{\strat \in \nasheq(\calG) \mid \rho(\strat)\models\phi}.$
 It is sometimes useful to consider a concurrent game that is modified so that no costs are incurred in it. We call such a game a \emph{cost-free game}. Where $\calG$ is a game, let $\calG^{\mathbf 0}$ denote the game that is the same as $\calG$ except that the cost function $\calC^{\mathbf 0}$ of $\calG^{\mathbf 0}$ is such that $\calC_i^{\mathbf 0}(s,\acp) = 0$ for all $i \in \calN$, $s \in \calS$, and $\acp \in \vac$. 
 Given this, the following is readily established (cf., \cite{DBLP:journals/ai/GutierrezNPW20}):
  
  \begin{theorem}\label{thm:ne}
    Given a game $\calG$, the problem of checking
    whether $\nasheq(\calG^{\mathbf 0}) \neq \emptyset$ is \textsc{2ExpTime}-complete.
  \end{theorem}

The notion of Nash equilibrium is closely related to the concept of beneficial deviations. Given how preferences are defined in this study, it will be useful to introduce terminology that captures the potential deviations that agents may have \cite{harrenstein:2017a}. Firstly, given a game $\calG$, we say that a strategy profile $\strat^1 \in \Sigma$ is \textit{distinguishable} from another strategy profile $\strat^2 \in \Sigma$ if $\rho(\strat^1,\calG) \neq \rho(\strat^2,\calG)$. Then, for an agent $i$, a strategy profile $\strat$, and an alternative strategy $\sigma_i' \neq \sigma_i$, we say that $\sigma_i'$ is an \textit{initial deviation} for agent $i$ from strategy profile $\strat$, written $\strat \to_i (\strat_{-i},\sigma_i')$, if we have $i \in W(\strat) \Rightarrow i \in W(\strat_{-i},\sigma_i')$ and strategy profile $\strat$ is distinguishable from $(\strat_{-i},\sigma_i')$.

\section{Taxation Schemes}
\label{sec:incentives}

We now introduce a model of incentives for concurrent games. For incentives to work, they clearly must appeal to an agent's preferences $\prefeq_i$. As we saw above, incentives for our games are defined with respect to both goals and costs: an agent's primary desire is to see their goal achieved -- the desire to minimise costs is strictly secondary to this. We will assume that we cannot change agents' goals: they are assumed to be fixed and immutable. It follows that any incentives we offer an agent to alter their behaviour must appeal to the costs incurred by that agent. Our basic model of incentives assumes that we can alter the cost structure of a game by imposing \emph{taxes}, which depend on the collective actions that agents choose in different states. Taxes may increase an agent's costs, influencing their preferences and rational choices.

Formally, we model static taxation schemes as functions $\tau : \calS \x \vac \to \Real_+^n.$ A static taxation scheme $\tau$ imposed on a game $\calG = (\calS,\calN,\ac_1, \ldots,\ac_n,\calT,\calC,\calL,s_0, \gamma_1, \ldots, \gamma_n)$ will result in a new game, which we denote by\[\calG^\tau = (\calS,\calN,\ac_1, \ldots,\ac_n,\calT,\calC^\tau,\calL,s_0, \gamma_1, \ldots, \gamma_n),\] which is the same as $\calG$ except that the cost function $\calC^\tau$ of $\calG^\tau$ is defined as $\calC^\tau(s,\acp) = \calC(s,\acp) + \tau(s,\acp).$ Similarly, we write $\calA^{\tau}$ to denote the arena with modified cost function $\calC^{\tau}$ associated with $\calG^{\tau}$ and $u_i^{\tau}(\rho)$ to denote the utility function of agent $i$ over run $\rho$ with the modified cost function $\calC^{\tau}$. Given $\calG$ and a taxation scheme $\tau$, we write $\rho_1 \prefeq_i^{\tau} \rho_2$ iff $u_i^{\tau}(\rho_1) \geq u_i^{\tau}(\rho_2)$. The indifference relations $\sim_i^{\tau}$ and strict preference relations $\pref_i^{\tau}$ are defined analogously.

The model of static taxation schemes has the advantage of simplicity, but it is naturally limited in the range of behaviours it can incentivise---particularly with respect to behaviours $\Upsilon$ expressed as LTL formulae. To overcome this limitation, we therefore introduce a \emph{dynamic} model of taxation schemes. This model essentially allows a designer to impose taxation schemes that can choose to tax the same action in different amounts, depending on the history of the run to date. A very natural model for dynamic taxation schemes is to describe them using a finite state machine with output---the same approach that we used to model strategies for individual agents. Formally, a \textit{dynamic taxation scheme} $T$ is defined by a tuple $T = (Q_T,next_T,do_T,q_T^0)$ where $Q_T$ is a finite set of taxation machine states, $next_T : Q_T \x \ac_1 \x \cdots \x \ac_n \to Q_T$ is the transition function of the machine, $q_T^0 \in Q_T$ is the initial state, and $do_T : Q_T \to (\calS \x \vac \to \Real_+^n)$ is the output function of the machine.
With this, let $\mathscr{T}$ be the set of all dynamic taxation schemes for a game $\calG$. As a run unfolds, we think of the taxation machine being executed alongside the strategies. At each time step, the machine outputs a static taxation scheme, which is applied at that time step only, with $do_T(q_T^0)$ being the initial taxation scheme imposed.

When we impose dynamic taxation schemes, we no longer have a simple transformation $\calG^\tau$ on games as we did with static taxation schemes $\tau$. Instead, we define the effect of a taxation scheme with respect to a run $\rho$. Formally, given a run $\rho$ of a game $\calG$, a dynamic taxation scheme $T$ induces an infinite sequence of static taxation schemes, which we denote by $t(\rho,T)$. We can think of this sequence as a function $t(\rho,T) : \Nat \to (\calS \x \vac \to \Real_+^n)$.
We denote the cost of the run $\rho$ in the presence of a dynamic taxation scheme $T$ by $\calC^T(\rho)$:
$$
\calC^T(\rho) = 
\liminf_{u\to\infty}
\frac{1}{u}\sum_{v=0}^{u} \calC(\rho,v) +
\underbrace{t(\rho,T)(v)(s(\rho,v),\acp(\rho,v))}_{(*)}
$$
The expression $(*)$ denotes the vector of taxes incurred by the agents as a consequence of performing the action profile which they chose at time step $v$ on the run $\rho$. The cost $\calC_i^T(\rho)$ to agent $i$ of the run $\rho$ under $T$ is then given by the $i$-th component of $\calC^T(\rho)$.

\begin{example}\label{example:robots}

 Two robots are situated in a grid world (Figure \ref{example:applerobots}), where atomic propositions represent events where a robot picks up an apple (label $a_{ij}$ represents agent $i$ picking up apple $j$), has delivered an apple to the basket (label $b_i$ represents agent $i$ delivering an apple to the basket), or where the robots have crashed into each other (label $c$). 
 Additionally, suppose that both robots are programmed with LTL goals $\gamma_1 = \gamma_2 = \always \neg c$. In this way, the robots are not pre-programmed to perform specific tasks, and it is therefore the duty of the \textit{principal} to design taxes that motivate the robots to perform a desired function, e.g., pick apples and deliver them to the basket quickly. Because the game is initially costless, there is an infinite number of Nash equilibria that could arise from this scenario and it is by no means obvious that the robots will choose one in which they perform the desired function. Hence, the principal may attempt to design a taxation scheme to eliminate those that do not achieve their objective, thus motivating the robots to collect apples and deliver them to the basket. Clearly, using dynamic taxation schemes affords the principal more control over how the robots should accomplish this than static taxation schemes.
\end{example}


\section{Nash Implementation}

We consider the scenario in which a principal, who is external to the game, has a particular goal that they wish to see satisfied within the game; in a general economic setting, the goal might be intended to capture some principle of social welfare, for example. In our setting, the goal is specified as an LTL formula $\Upsilon$, and will typically represent a desirable system/global behaviour. The principal has the power to influence the game by choosing a taxation scheme and imposing it upon the game. Then, given a game $\calG$ and a goal $\Upsilon$, our primary question is whether it is possible to design a taxation scheme $T$ such that, assuming the agents, individually and independently, act rationally (by choosing strategies $\strat$ that collectively form a Nash equilibrium in the modified game), the goal $\Upsilon$ will be satisfied in the run $\rho(\strat)$ that results from executing the strategies $\strat$. In this section, we will explore two ways of interpreting this problem.

\vspace*{1ex}\noindent\textbf{E-Nash Implementation:} A goal $\Upsilon$ is \emph{E-Nash implemented} by a taxation scheme $T$ in $\calG$ if there is a Nash equilibrium strategy profile $\strat$ of the game $\calG^T$ such that $\rho(\strat) \models \Upsilon$. The notion of E-Nash implementation is thus analogous to the E-Nash concept in rational verification~\cite{gutierrezetal:2017a,gutierrez:2020a}. Observe that, if the answer to this question is ``yes'' then this implies that the game $\calG^T$ has at least one Nash equilibrium. Let us define the set $\eni$ to be the set of taxation schemes $T$ that E-Nash implements $\Upsilon$ in~$\calG$:
\[
\eni = 
    \set{T \in \mathscr{T} \mid \nasheq_\Upsilon(\calG^T) \neq \emptyset} \ .
\]
The obvious decision problem is then as follows:
\begin{quote}
\underline{\textsc{E-Nash Implementation}}:\\
\emph{Given}: Game $\calG$, LTL goal $\Upsilon$.\\
\emph{Question}: Is it the case that $\eni \neq \emptyset$?
\end{quote}
This decision problem proves to be closely related to the \textsc{E-Nash} problem~\cite{gutierrezetal:2017a,gutierrez:2020a}, and the following result establishes its complexity:

\begin{theorem}\label{thm:e-nashcomplexity}
\textsc{E-Nash Implementation} is \textsc{2ExpTime}-complete, even when $\mathscr{T}$ is restricted to static taxation schemes.
\end{theorem}

\begin{proof}
For membership, we can check whether $\Upsilon$ is satisfied on any Nash equilibrium of the cost-free concurrent game $\calG^{\mathbf 0}$ obtained from $\calG$ by effectively removing its cost function using a static taxation scheme which makes all costs uniform for all agents. This then becomes the \textsc{E-Nash} problem, known to be \textsc{2ExpTime}-complete. The answer will be ``yes'' iff $\Upsilon$ is satisfied on some Nash equilibrium of $\calG^{\mathbf 0}$; and if the answer is ``yes'', then observing that $\nasheq(\calG^{T})  \subseteq \nasheq(\calG^{\mathbf 0})$ for all taxation schemes $T \in \mathscr{T}$ \cite{wooldridge:2013a}, the given LTL goal~$\Upsilon$ can be E-Nash implemented in $\calG$. For hardness, we can reduce the problem of checking whether a cost-free concurrent game $G$ has a Nash equilibrium (Theorem~\ref{thm:ne}). Simply ask whether $\Upsilon = \top$ can be E-Nash implemented in $\calG^{\mathbf 0}$.

For the second part of the result, observe that the reduction above only involves removing the costs from the game and checking the answer to \textsc{E-Nash}, which can be done using a simple static taxation scheme. Hardness follows in a similar manner.
\end{proof}

\vspace*{1ex}\noindent\textbf{A-Nash Implementation:} The universal counterpart of E-Nash implementation is \emph{A-Nash Implementation}.
We say that $\Upsilon$ is \emph{A-Nash implemented} by $T$ in $\calG$ if we have both 1) $\Upsilon$ is E-Nash implemented by $T$ in game~$\calG$; and 2) $\nasheq(\calG^T) = \nasheq_\Upsilon(\calG^T)$. 
We thus define $\ani$ as follows:
\[
\ani = 
    \set{T \in \mathscr{T} 
        \mid \nasheq(\calG^T) = \nasheq_\Upsilon(G^T) \neq \emptyset} 
\]
The decision problem is then:
\begin{quote}
\underline{\textsc{A-Nash Implementation}}:\\
\emph{Given}: Game $\calG$, LTL goal $\Upsilon$.\\
\emph{Question}: Is it the case that $\ani \neq \emptyset$?
\end{quote}

\begin{figure}
        \begin{subfigure}{.5\textwidth}
            \centering
            \scalebox{0.9}{\begin{tikzpicture}[->,>=stealth',shorten >=1pt,auto,node distance=1cm, thick, initial text = ] 
	
 \node[state, initial, initial where=above] (s0) {$s_0$}; 
	\node[state, label={above:$\{p,q\}$}] (s1) at (-3.5,0){$s_1$}; 
	\node[state, label={above:$\{p\}$}] (s2) at (3.5,0) {$s_2$}; 
	\node[state, label={left:$\{q\}$}] (s3) at (0,-2) {$s_3$};
	
	\draw   (s0) edge [bend right, above, align=center] node {$\acp_{2}:(2,0)$\\$\acp_{3}:(0,2)$} (s1)   	
	        (s0) edge [bend left, above] node {$\acp_{1}:(0,0)$} (s2)
	        (s0) edge [bend left] node [pos=0.5] {$\acp_{4}:(2,2)$} (s3)
	        (s1) edge node [pos=0.5, below] {$\vac:(0,0)$} (s0)
	        (s2) edge node [pos=0.5, below] {$\vac:(0,0)$} (s0)
	        (s3) edge [bend left] node {$\vac : (0,0)$} (s0);
\end{tikzpicture}}
            \caption{}
            \label{fig:dynamicvsstatic2}
        \end{subfigure}
        \begin{subfigure}{.4\textwidth}
            \centering
            \scalebox{0.9}{\begin{tikzpicture}[->,>=stealth',shorten >=1pt,auto,node distance=1cm, thick, initial text = ] 

	\node[state, initial, initial where=above, label={below:$\vac:(0,0)$}] (q0)   {$q_T^0$}; 
	\node[state, label={below:$\vac:(0,0)$}] (q1) at (2.5,0) {$q_T^1$}; 
	\node[state, label={below:$\vac:(3,3)$}] (q2) at (-2.5,0) {$q_T^2$}; 
	
	\draw   (q0) edge [bend left] node [pos=0.5, above] {$\acp_2,\acp_3$} (q1)
	        (q0) edge node [pos=0.4, above] {$\acp_1,\acp_4$} (q2)
                (q1) edge [bend left] node [pos=0.5, above] {$\vac$} (q0)
	        (q2) edge [loop above] node [pos=0.5, left=1mm] {$\vac$} (q2);
\end{tikzpicture}}
            \vspace{1.5em}
            \caption{}
            \label{fig:dynamicincentive}
        \end{subfigure}
        \caption{(a): A two-agent concurrent game $\calG$ with action sets $\ac_1 = \set{a,b}$ and $\ac_2 = \set{c,d}$ and goals $\gamma_1 = \gamma_2 = \always \eventually p$, where we let $\acp_{1} = (a,c),\acp_{2} = (a,d),\acp_{3} = (b,c),\acp_{4} = (b,d)$. Cost vectors associated with sets denote that all action profiles within the set are assigned those costs. (b): A dynamic taxation scheme that could be imposed on the agents in the game from (a). Labels below the states represent a static taxation scheme that applies a uniform tax for all agents and all action profiles.}
\end{figure}
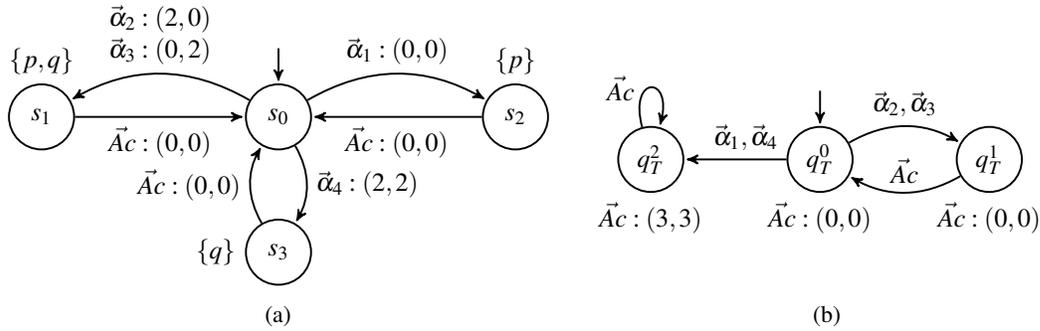

The following result shows that, unlike the case of E-Nash implementation, dynamic taxation schemes are \textit{strictly more powerful} than static taxation schemes for A-Nash implementation. It can be verified that the game in Figure \ref{fig:dynamicvsstatic2}, the taxation scheme in Figure \ref{fig:dynamicincentive}, and the principal's goal being $\Upsilon = G(p \leftrightarrow q)$ are witnesses to this result (see Appendix for the full proof):

\begin{proposition}\label{observation:dynamic>static}
There exists a game $\calG$ and an LTL goal $\Upsilon$ such that $\ani \neq \emptyset$, but not if $\mathscr{T}$ is restricted to static taxation schemes.
\end{proposition}

Before proceeding with the \textsc{A-Nash Implementation} problem, we will need to introduce some additional terminology and concepts, beginning first with deviation graphs, paths, and cycles. A \textit{deviation graph} is a directed graph $\Gamma = (\calV, E)$, where $\calV \subseteq \Sigma$ is a set of nodes which represent strategy profiles in $\Sigma$ and $E \subseteq \set{(\strat,\strat') \in \calV \times \calV \mid \strat \to_i \strat' \mbox{ for some } i \in \calN}$ is a set of directed edges between strategy profiles that represent initial deviations. Additionally, we say that a dynamic taxation scheme $T$ \textit{induces} a deviation graph $\Gamma = (\calV,E)$ if for every $(\strat,\strat') \in \calV \times \calV$, it holds that $\strat' \pref_i^T \strat$ for some $i \in \calN$ if and only if $(\strat,\strat') \in E$. In other words, if the edges in a deviation graph precisely capture all of the beneficial deviations between its nodes under $T$, then the deviation graph is said to be induced by $T$.\footnote{This definition implies that a taxation scheme may induce many possible deviation graphs in general, depending on the nodes selected to be part of the graph.} Then, a \textit{deviation path} is simply any path $P = (\strat^1, \ldots,\strat^m)$ within a deviation graph $\Gamma$ where $(\strat^j,\strat^{j+1}) \in E$ for all $j \in \set{1,\ldots,m-1}$. 

Because the principal is only able to observe the actions taken by the agents and not their strategies directly, any taxation scheme that changes the cost of some strategy profile $\strat$ will also change the cost of all strategy profiles that are indistinguishable from $\strat$ by the same amount. 
This naturally suggests that we modify the concept of a deviation path to take indistinguishability into account. To this end, we say that a sequence of runs $P_o = (\rho^1,\rho^2,\ldots,\rho^m)$ is an \textit{observed deviation path} in a deviation graph $\Gamma = (\calV,E)$ if there exists an \textit{underlying tuple} $(\strat^1, \strat^2, \ldots, \strat^m)$ such that for all $j \in \set{1,\ldots,m}$, it holds that 1) $\rho^j = \rho(\strat^j)$, and 2) if $j < m$, then $(\strat^j,\strat^{j+1'}) \in E$ for some $\strat^{j+1'}$ such that $\rho(\strat^{j+1'}) = \rho(\strat^{j+1})$. 
Then, a \textit{deviation cycle} is a deviation path $(\strat^1,\ldots,\strat^m)$ where $\rho(\strat^1) = \rho(\strat^m)$.
A deviation path $P = (\strat^1, \strat^2, \ldots, \strat^m)$ is said to \textit{involve} an agent $i$ if $\strat^j \to_i \strat^{j+1}$ for some $j \in \set{1,\ldots,m-1}$ and similarly, an observed deviation path $P_o$ in a deviation graph involves agent $i$ if the analogous property holds for all of its underlying sets. Given a game $\mathcal{G}$ and a set of strategy profiles $X$, a taxation scheme $T$ \textit{eliminates} $X$ if $NE(\mathcal{G}^T) \cap X = \emptyset$. Finally, a set of strategy profiles $X$ is said to be \textit{eliminable} if there exists a taxation scheme that eliminates it. With this, we can characterise the conditions under which a finite set of strategy profiles is eliminable:

\begin{proposition}\label{prop:elimset}
    Let $\calG$ be a game and $X \subset \Sigma$ be a finite set of strategy profiles in $\calG$. Then, $X$ is eliminable if and only if there exists a finite deviation graph $\Gamma = (\calV, E)$ that satisfies the following properties: 1) For every $\strat \in X$, there is some $\strat' \in \calV$ such that $(\strat,\strat') \in E$; and 2) Every deviation cycle in $\Gamma$ involves at least two agents.
\end{proposition}

\begin{proof}[Proof Sketch.]
The forward direction follows by observing that if all deviation graphs fail to satisfy at least one of the two properties, then every deviation graph will either fail to eliminate some $\strat \in X$ if induced, or will not be inducible by any dynamic taxation scheme. The backward direction can be established by constructing a dynamic taxation scheme $T^{\Gamma}$ that induces a deviation graph $\Gamma$ satisfying the two properties. Using these properties, it follows that $T^{\Gamma}$ eliminates $X$.
\end{proof}

To conclude our study of dynamic taxation schemes, we present a characterisation of the A-Nash implementation problem.\footnote{Note that, in general, Proposition \ref{prop:elimset} cannot be directly applied to Theorem \ref{thm:a-nashcharacterisation}, because it is assumed that the set to be eliminated is finite, whereas $\nasheq_{\neg\Upsilon}(\calG^{\boldsymbol{0}})$ is generally infinite. However, this can be reconciled if some restriction is placed on the agents' strategies so that $\Sigma$ is finite, which is the case in many game-theoretic situations of interest, e.g., in games with memoryless, or even bounded memory, strategies -- both used to model bounded rationality.}

\begin{theorem}\label{thm:a-nashcharacterisation}
Let $\calG$ be a game and $\Upsilon$ be an LTL formula. Then $\ani \neq \emptyset$ if and only if the following conditions hold:
\begin{enumerate}
    \item $\eni \neq \emptyset$;
    \item $\nasheq_{\neg\Upsilon}(\calG^{\boldsymbol{0}})$ is eliminable.
\end{enumerate}
\end{theorem}
 \begin{proof}
 For the forward direction, it follows from the definition of the problem that if $\eni = \emptyset$, then $\ani = \emptyset$. Moreover, it is also clear that if $\nasheq_{\neg\Upsilon}(\calG^{\boldsymbol{0}})$ is not eliminable, then it is impossible to design a (dynamic) taxation scheme such that only good equilibria remain in the game and hence, $\ani = \emptyset$.

 For the backward direction, suppose that the two conditions hold and let $T$ be a taxation scheme that only affects the limiting-average costs incurred by agents under strategy profiles in $\nasheq_{\neg\Upsilon}(\calG^{\boldsymbol{0}})$, and eliminates this set. Such a taxation scheme is guaranteed to exist by the assumption that condition $(2)$ holds and because it is known that no good equilibrium is indistinguishable from a bad one. Now consider a static taxation scheme $\tau$ such that $c_i(s,\acp) + \tau_i(s,\acp) = \hat{c}$ for all $i \in \calN$, $(s,\acp) \in \calS \times \vac$, and some $\hat{c} \geq \max_{i \in \calN} c_i^*$. Combining $\tau$ with $T$ gives us a taxation scheme $T^*$ such that for each state $q \in Q_{T^*} = Q_T$ and $(s,\acp) \in \calS \times \vac$, we have $do_{T^*}(q)(s,\acp) = do_{T}(q)(s,\acp) + \tau(s,\acp)$.  Now, because $T$ eliminates $\nasheq_{\neg\Upsilon}(\calG^{\boldsymbol{0}})$, and $\nasheq(\calG^\tau) = \nasheq(\calG^{\boldsymbol{0}})$, it follows that $T^*$ eliminates $\nasheq_{\neg\Upsilon}(\calG^{\boldsymbol{0}})$. Finally, note that because the satisfaction of an LTL formula on a given run is solely dependent on the run's trace, it follows that all good equilibria, i.e., strategy profiles in $\nasheq_{\Upsilon}(\calG^{\boldsymbol{0}})$, are distinguishable from all bad equilibria, so we have $\nasheq_{\Upsilon}(\calG^{\boldsymbol{0}}) \cap \nasheq(\calG^{T^*}) \neq \emptyset$.
\end{proof}

It is straightforward to see that \textsc{A-Nash Implementation} is 2EXPTIME-hard via a simple reduction from the problem of checking whether a Nash equilibrium exists in a concurrent game -- simply ask if the formula $\top$ can be A-Nash implemented in $\calG^{\boldsymbol{0}}$. However, it is an open question whether a matching upper bound exists and we conjecture that it does not. This problem is difficult primarily for two reasons. Firstly, it is well documented that Nash equilibria may require infinite memory in games with lexicographic $\omega$-regular and mean-payoff objectives \cite{chatterjee:2005a}, and the complexity of deciding whether a Nash equilibrium even exists in games with our model of preferences has yet to be settled \cite{gutierrez:2020a}. Secondly, Theorem \ref{thm:a-nashcharacterisation} and Proposition \ref{prop:elimset} suggest that unless the strategy space is restricted to a finite set, a taxation scheme that A-Nash implements a formula may require reasoning over an infinite deviation graph, and hence require infinite memory. Nevertheless, our characterisation under such restrictions provides the first step towards understanding this problem in the more general setting.

\section{Related Work and Conclusions}
\label{sec:conclusion}

This work was motivated by~\cite{wooldridge:2013a}, and based on that work, presents four main contributions: the introduction of {\em static and dynamic} taxation schemes as an extension to concurrent games expanding the model in (one-shot) Boolean games~\cite{wooldridge:2013a, harrenstein:2014a, harrenstein:2017a}; a study of the \textit{complexity} of some of the most relevant computational decision problems building on previous work in rational verification~\cite{gutierrezetal:2017a,DBLP:journals/ai/GutierrezNPW20,gutierrez:2020a}; evidence (formal proof) of the strict \textit{advantage of dynamic taxation schemes} over static ones, which illustrates the role of not just observability but also \textit{memory} to a principal's ability to (dis)incentivise certain outcomes \cite{grossman1992analysis,holmstrom1982moral}; and a full characterisation of the \textit{eliminability of sets of strategy profiles} under dynamic taxation schemes and the A-Nash implementation problem. 

The incentive design problem has been studied in many different settings, and \cite{ratliff2019perspective} group existing approaches broadly into those from the economics, control theory, and machine learning communities. However, more recent works in this area adopt multi-disciplinary methods such as automated mechanism design~\cite{parkes2010dynamic,mguni2019efficient,shen2019automated,balaguer2022good}, which typically focus on the problem of constructing incentive-compatible mechanisms to optimise a particular objective such as social welfare. Other approaches in this area reduce mechanism design to a program synthesis problem~\cite{narayanaswamiautomating} or a satisfiability problem for quantitative strategy logic formulae~\cite{maubert2021strategic,mittelmann2022automated}. The notion of dynamic incentives has also been investigated in (multi-agent) learning settings~\cite{centeno2011using,mguni2019coordinating,ratliff2020adaptive,yang2021adaptive,elbarbari2022framework}. These works focus solely on adaptively modifying the rewards for quantitative reward-maximising agents. In contrast, our model of agent utilities more naturally captures fundamental constraints on the principal's ability to (dis)incentivise certain outcomes due to the lexicographic nature of agents' preferences \cite{brauning2017lexicographic}.

Another area closely related to incentives is that of norm design \cite{mahmoud2014review}. Norms are often modelled as the encouragement or prohibition of actions that agents may choose to take by a regulatory agent. 
The most closely related works in this area are those of~\cite{huang2016normative,perelli2019enforcing,alechina2022automatic}, who study the problem of synthesising dynamic norms in different classes of concurrent games to satisfy temporal logic specifications. Whereas norms in these frameworks have the ability to \textit{disable} actions at runtime, our model confers only the power to \textit{incentivise} behaviours upon the principal. Finally, other studies model norms with violation penalties, but differ from our work in how incentives, preferences, and strategies are modelled~\cite{cardoso2009adaptive,bulling2016norm,dellanna2020runtime}.

In summary, a principal's ability to align self-interested decision-makers' interests with higher-order goals presents an important research challenge for promoting cooperation in multi-agent systems. The present study highlights the challenges associated with incentive design in the presence of constraints on the kinds of behaviours that can be elicited, makes progress on the theoretical aspects of this endeavour through an analysis of taxation schemes, and suggests several avenues for further work. Promising directions include extensions of the game model to probabilistic/stochastic or learning settings, finding optimal complexity upper bounds for the A-Nash implementation problems, and consideration of different formal models of incentives. We expect that this and such further investigations will positively contribute to our ability to develop game-theoretically aware incentives in multi-agent systems.

\bibliographystyle{eptcs}
\bibliography{refs}

\clearpage

\setcounter{proposition}{0}
\setcounter{theorem}{2}
\setcounter{lemma}{0}
\setcounter{corollary}{0}

\clearpage

\section{Supplementary Material}

\begin{proposition}
There exists a game $\calG$ and an LTL goal $\Upsilon$ such that $\ani \neq \emptyset$, but not if $\mathscr{T}$ is restricted to static taxation schemes.
\end{proposition}

\begin{proof}
Consider the concurrent game $\calG$ in Figure \ref{fig:dynamicvsstatic2}. Intuitively, both agents desire to always eventually visit either $s_1$ or $s_2$. Suppose that the principal's objective is $\Upsilon = G(p \leftrightarrow q)$, i.e., they would like the agents to never visit $s_2$ or $s_3$.
Firstly, observe that there is no \textit{static} taxation scheme which can A-Nash implement $\Upsilon$, as any modification to the costs of the game will not eliminate any Nash equilibria where the agents visit $s_2$ or~$s_3$ a finite number of times. This is due to the prefix-independence of costs in infinite games with limiting-average payoffs \cite{ummels2011complexity}.
However, the dynamic taxation scheme depicted in Figure \ref{fig:dynamicincentive} A-Nash implements $\Upsilon$. To see this, observe that for any strategy profile that visits $s_2$ or $s_3$ a finite number of times, there exists a deviation for some agent to ensure that $s_2$ and $s_3$ are never visited. Such a deviation will result in all agents $i \in \set{1,2}$ satisfying their goals $\gamma_i$ and strictly reducing their average costs from at least $c_i^*+1$ to some value strictly below this. This constitutes a beneficial deviation and hence, there is no Nash equilibrium under $T$ that does not satisfy $\Upsilon$. Moreover, any strategy profile $\strat$ that leads to the sequence of states $s(\rho(\strat),0:) = (s_0 s_1)^{\omega}$ is a Nash equilibrium of $\calG^T$ and hence goal $\Upsilon$ is A-Nash implemented by $T$ in this game. 
\end{proof}

\begin{proposition}
    Let $\calG$ be a game and $X \subset \Sigma$ be a finite set of strategy profiles in $\calG$. Then, $X$ is eliminable if and only if there exists a finite deviation graph $\Gamma = (\calV, E)$ that satisfies the following properties: 1) For every $\strat \in X$, there is some $\strat' \in \calV$ such that $(\strat,\strat') \in E$; and 2) Every deviation cycle in $\Gamma$ involves at least two agents.
\end{proposition}

\begin{proof}
For the forward direction, suppose that there is no deviation graph $\Gamma$ satisfying both properties (1) and (2) in the statement. Then, for all deviation graphs $\Gamma$, either for some $\strat \in X$, there is no $\strat' \in \calV$ such that $(\strat,\strat') \in \calV$, or there is some deviation cycle in $\Gamma$ involving only one agent. Now consider any deviation graph $\Gamma = (\calV,E)$, where $\calV = X \cup \set{\strat' \mid \strat \to_i \strat' \mbox{ for some } \strat \in X \mbox{ and } i \in \calN}$. In the first case, it is clear that any taxation scheme that induces $\Gamma$ does not eliminate $\set{\strat}$ and hence $X$. In the second case, no taxation scheme can induce the deviation graph $\Gamma$. To see why, suppose for contradiction that some taxation scheme $T$ induces $\Gamma$ and let $i$ be the agent for which there is a deviation cycle $C = \set{\strat^1,\ldots,\strat^m}$ in $\Gamma$ involving only agent $i$. Then, we have $\strat^1 \pref_i^T \strat^2 \pref_i^T \ldots \pref_i^T \strat^m$ and by transitivity of the preference relation $\pref_i^T$, we can conclude that $\strat^1 \pref_i^T \strat^m$. However, by definition of a deviation cycle, $\strat^1$ and $\strat^m$ are indistinguishable, so agent $i$ will always receive the same utility under both $\strat^1$ and $\strat^m$, no matter what taxation scheme is imposed on them and hence, we have a contradiction. From this, we can conclude that every deviation graph that can be induced by a taxation scheme does not eliminate $X$ and hence, $X$ is not eliminable, proving this part of the statement. 

For the backward direction, assume that there is a deviation graph $\Gamma$ that satisfies both properties. Under this assumption, we will construct a dynamic taxation scheme $T$ that eliminates $X$. To assign the appropriate costs to different strategy profiles, we will make use of the lengths of deviation paths within $\Gamma$. For every $i \in \calN$, let $\ell_i$ denote the length of the longest observed deviation path in $\Gamma$ that involves only agent $i$. Additionally, for all $\strat \in \calV$, let $d_i(\rho(\strat))$ denote the length of the longest \textit{observed} deviation path in $\Gamma$ that starts from $\rho(\strat)$ and involves only $i$. The difference between these two quantities will serve as the basis for how much taxation an agent $i$ will incur for any given strategy profile in $\calV$. Observe that because it is assumed that no deviation cycle involves only one agent, both quantities are well-defined and finite for all agents and strategy profiles. Then, for a deviation graph $\Gamma$ and a run $\rho$, let $\indev{\rho}$ be the set of agents $i \in \calN$ for which there is some pair of strategy profiles $\strat, \strat' \in \calV$ such that we have both $(\strat,\strat') \in E_D$ and $\rho = \rho(\strat')$. In other words, $\indev{\rho}$ represents the set of agents who have an initial deviation from some other strategy profile in $\calV$ to one that generates the run $\rho$. With this, we would like to construct a dynamic taxation scheme such that for any strategy profile $\strat$, the following criteria are satisfied:
\begin{itemize}
    \item $C_i^T(\rho(\strat)) \geq (\ell_i - d_i(\rho(\strat)))\cdot (c_i^*+1) \qquad$ if $i \in \indev{\rho}$;
    \item $C_i^T(\rho(\strat)) = C_i(\rho(\strat)) \qquad \qquad \qquad \qquad\ $ otherwise.
\end{itemize}

Intuitively, the idea is to ensure that for every edge $(\strat,\strat') \in E$, the agent $i \in \calN$ for whom $\strat \to_i \strat'$ gets taxed by a significantly higher amount for choosing $\strat$ compared to when they choose $\strat'$. To see why it is possible to construct such a taxation scheme, first observe that if $\rho \neq \rho'$ for any two runs $\rho,\rho'$, then there is some dynamic taxation scheme that can distinguish between the two by simply tracing out the two runs up to the first point in which they differ and then branching accordingly. From this point onwards, the dynamic taxation scheme can then output static taxation schemes, which assign different limiting average costs to the agents according to the above criteria. Extending this approach to a taxation scheme that distinguishes between all unique runs generated by elements of $\calV$, it follows that there is a dynamic taxation scheme $T$ that satisfies the two criteria. Consequently, for all $(\strat,\strat') \in E$, it follows that $\strat' \pref_i^T \strat$ because $\rho(\strat) \neq \rho(\strat')$ by definition of the initial deviation relation $\to_i$. Moreover, because it is assumed that no deviation cycle involves only one agent, $T$ gives rise to a strict total ordering $\pref_i^T$ on the elements of $\calV$ for each $i \in \calN$. Finally, by property $(1)$, it holds that for every $\strat \in X$, some agent has a beneficial deviation from $\strat$ to another $\strat' \in \calV$ under $T$ and hence, $T$ eliminates $X$.
\end{proof}

\end{document}